\documentclass[journal,12pt]{IEEEtran}
\usepackage{filecontents}
\usepackage[noadjust]{cite}
\usepackage{caption}
\usepackage{amsfonts}
\usepackage{subcaption}
\usepackage[dvips]{color, graphicx}
\usepackage[table,xcdraw]{xcolor}
\definecolor{tabgray}{RGB}{200,200,200}
\usepackage{multirow}
\usepackage{multicol}
\usepackage{makecell}

\usepackage{graphicx}
\usepackage{threeparttable}
\usepackage{pgfplotstable}
\usepackage{mathrsfs}
\usepackage{stmaryrd}
\usepackage{cprotect}
\usepackage[cmex10]{amsmath}
\usepackage{algorithm}
\usepackage{algorithmicx}
\usepackage{amssymb}
\usepackage{amsthm}
\usepackage{algpseudocode}
\usepackage{pgfplotstable}
\usepackage{array}
\usepackage{booktabs}
\usepackage{times}
\usepackage[normalem]{ulem} % LL added, apr. 3, used to cross out sentences
\long\def\symbolfootnote[#1]#2{\begingroup%
\def\thefootnote{\fnsymbol{footnote}}\footnote[#1]{#2}\endgroup}

\begin{document}
\title{Polar Codes with Memory}
\author{
Wenyue~Zhou,~\IEEEmembership{Student~Member,~IEEE,} Qiang~Liu,~\IEEEmembership{Student~Member,~IEEE,}
Yifei~Shen,~\IEEEmembership{Student~Member,~IEEE,} Xiaofeng~Zhou,~\IEEEmembership{Student~Member,~IEEE,}
Chuan~Zhang,~\IEEEmembership{Member,~IEEE,} Yaohua~Xu,~\IEEEmembership{Member,~IEEE} and~Liping~Li,~\IEEEmembership{Member,~IEEE.}
\vspace{-7mm}
\thanks{This work was supported in part by the National Natural Science Foundation
of China through grant 61501002, in part by the Natural Science Project of
Ministry of Education of Anhui through grant KJ2015A102, and
in part by the Talents Recruitment Program of Anhui University.}
\thanks{Wenyue Zhou, Yaohua Xu and Liping Li are with the Key Laboratory of Intelligent Computing and Signal Processing,
Ministry of Education,
Anhui University, Hefei, China (liping\_li@ahu.edu.cn).}
\thanks{
Qiang Liu, Yifei Shen, Xiaofeng Zhou and Chuan Zhang are with the National Mobile Communications Research Laboratory, Southeast University,
Nanjing, China (chzhang@seu.edu.cn)}
%\IEEEauthorblockN{Wenyue Zhou}
%\IEEEauthorblockA{
%Key Laboratory of Intelligent Computing and\\
%Signal Processing of the Ministry of Education of China\\
%Anhui University, China,\\
%Email: p$17201059$@stu.ahu.edu.cn}
%\and
%\IEEEauthorblockN{Liping Li}
%\IEEEauthorblockA{
%Key Laboratory of Intelligent Computing and\\
%Signal Processing of the Ministry of Education of China\\
%Anhui University, China,\\
%Email: liping\_li@ahu.edu.cn}
}

\maketitle

\begin{abstract}
Polar codes with memory (PCM) are proposed in this paper: a pair of consecutive code blocks containing a controlled number of mutual information bits. The shared mutual information bits of the succeeded block can help the failed block to recover. The underlying polar codes can employ any decoding scheme such as the successive cancellation (SC) decoding (PCM-SC), the belief propagation (BP) decoding (PCM-BP), and the successive cancellation list (SCL) decoding (PCM-SCL). The analysis shows that the packet error rate (PER) of PCM decreases to the order of PER squared while maintaining the same complexity as the underlying polar codes. Simulation results indicate that for PCM-SC, the PER is comparable to (less than 0.3 dB) the stand-alone SCL decoding with two lists  for the block length $N=256$. The PER of PCM-SCL with $L$ lists can match that of the stand-alone SCL decoding with $2L$ lists. Two hardware decoders for PCM are also implemented: the in-serial (IS) decoder and the low-latency interleaved (LLI) decoder. For $N=256$, synthesis results show that in the worst case, the latency of the PCM LLI decoder is only $16.1\%$ of the adaptive SCL decoder with $L=2$, while the throughput is improved by 13 times compared to it.

\end{abstract}


\section{Introduction}\label{sec_ref}
Polar codes invented by Ar\i kan \cite{arikan_iti09} have been proven to be a coding scheme that can achieve the capacity of symmetric binary-input discrete memoryless channels (B-DMCs) with low complexity of encoding and successive cancellation (SC) decoding.
Nevertheless, on account of the insufficient polarization,
the error-correcting performance of moderate length polar codes under SC decoding is unsatisfactory \cite{urbanke_isit09, vardy_it15}.
To acquire better finite-length performance, successive cancellation list (SCL) decoding was proposed in \cite{niu_itc11, niu_itc13, vardy_it15} and
it is comparable to low-density parity-check (LDPC) codes in terms of error-correcting performance.
%its performance was shown to be comparable to low-density parity-check (LDPC) codes.
Belief propagation (BP) was an alternative decoding algorithm \cite{arikan_icl08, urbanke_isit09, eslami_allerton10} over the factor graph of polar codes. It has better performance than the SC decoding and supports parallel decoding.
But the bit error rate (BER) performance of polar codes with BP decoding is still inferior to the SCL decoding (shown in this paper).

In this paper, a new construction scheme of polar codes is proposed by sharing a controlled number of
information bits between two consecutive encoding blocks. The input stream is divided into an odd stream and an even
stream. In the encoding process, the corresponding odd and even blocks share a fraction of information bits, which
are called \textit{mutual information bits} in this paper. Cyclic redundancy check (CRC) bits are attached to the information bits in each block.
The encoding of these two blocks can be done sequentially or in parallel.
In the decoding process, only when one of the pair is decoded correctly, the succeeded block can
provide the estimations of the mutual information bits to the failed block.
With a proper design of the positioning of the mutual information bits, the failed block can be
recovered with another round of decoding.

Since the two consecutive blocks share mutual information bits, it is like there
is some memory in the encoding process. Therefore, we call
the proposed scheme polar codes with memory (PCM) to differentiate this scheme from the traditional polar encoding scheme.
In addition, this scheme can be directly extended to $m$ ($m>2$) blocks. Based on this, a general PCM is proposed in this paper, which reduces the effective code rate loss while maintaining the same order of PER, compared with the direct extension of PCM.
Analysis shows that the packet error rate (PER) of PCM is only square of that of the underlying polar codes.
Note that this great performance improvement comes with a complexity the same as the underlying polar codes.
The decoding of PCM can be implemented by the SC, BP, or SCL decoding. In other words, the decoding complexity of PCM is the complexity of the underlying SC, BP, or SCL decoding.
For ease of description, PCM-SC-$2$ is used to refer to the PCM employing the SC decoding and two blocks sharing mutual information bits. Similarly, PCM-BP-$2$ and PCM-SCL-$2$ refer to two blocks sharing mutual information bits, each block employing the BP and SCL decoding, respectively. In addition, PCM-SC-$m$ ($m>3$) refers to the general PCM employing the SC decoding with $m$ blocks.

The simulation results show that the PER of PCM-SC-2
is only 0.3 dB away from the stand-alone SCL decoding with $L=2$ ($L$ being the list size) for the studied case in
the paper. The performance of PCM-BP-2 can achieve the same performance of the stand-alone SCL decoding with $L=2$. In addition, the performance of PCM-SCL-2 with $L$ lists matches the performance of the stand-alone SCL decoding with $2L$ lists. Two hardware architectures are also proposed in the paper: an in-serial (IS) architecture and a low-latency interleaved (LLI) architecture. Implementation results show that for the block length 256, the proposed LLI architecture for PCM with the SC decoding has lower latency and higher throughput compared to the adaptive SCL decoder ($L=4$) \cite{Sural01}.

The rest of the paper is organized as follows. Section \ref{sec_background} is on the basics of polar codes.
Section \ref{sec_polar_memory} introduces the proposed PCM scheme. Specifically,
Section \ref{sec_encoding_memory} introduces the encoding process of PCM, Section
\ref{sec_decoding_process} is about the corresponding decoding process, and in Section \ref{sec_mutual_bits},
the optimal strategy to position mutual information bits is proposed. The error performance
of PCM is analyzed in Section \ref{sec_error_analysis}. The application of a BP or SCL decoder
to PCM is introduced in Section \ref{sec_bp_decoding}. We also compare PCM
with Turbo codes in Section \ref{sec_comparision_turbo}.
Moreover, we extend the PCM to $m > 2$ blocks in Section \ref{sec_gen_polar}.
In Section \ref{sec_numerical}, the simulation results are provided to validate the proposed PCM.
In Section \ref{sec_hardware}, the hardware architectures of PCM decoding are implemented.
The concluding remarks are provided at the end.
\section{Preliminaries of Polar Codes}\label{sec_background}
Denote $v_1^N$ as an $N$-length vector $(v_1,...,v_N)$.
Let $W:~\mathcal{X}$ $\rightarrow$ $\mathcal{Y}$ denote a symmetric B-DMC, with the input alphabet $\mathcal{X}=\{0,1\}$, the
output alphabet $\mathcal{Y}$, and the channel transition probability $W(y|x)$, $x\in\mathcal{X}$, $y\in\mathcal{Y}$.
Let $N=2^n$ ($n \ge 1$) denote the block code length.
The  generator matrix of polar codes is $G$, which  is given by $G=B_NF^{\otimes n}$.
Here $B_N$ denotes the bit-reversal permutation matrix, $F=\bigl[\begin{smallmatrix} 1 & 0 \\ 1 & 1 \end{smallmatrix}\bigr]$, and $F^{\otimes n}$ represents the $n$-th Kronecker power of $F$ over the binary field $\mathbb{F}_2$.
The codewords $x_1^N$ can be obtained by $x_1^N=u_1^NG$, where $u_1^N$ is the source vector,
consisting of $K$ information bits and $N-K$ frozen bits (the fixed information in the source vector).
The codeword $x_1^N$ is transmitted over $N$ independent copies of $W$, written as $W^N$, with a transition probability $W^N(y_1^N|x_1^N)$.

Channel polarization process has two parts: channel combining and channel splitting. Channel combining is a phase that combines copies of $W$ in a recursive manner to produce a vector channel $W_N$, with $W_N(y_1^N|u_1^N) = W^N(y_1^N|x_1^N)$. Channel splitting is an operation splitting $W_N$ back into a set of $N$ binary-input channels $W_N^{(i)}$, $i \in \{1,2,...,N\}$. The
$i$-th such channel is called bit channel $i$ (meaning the channel that bit $i$ virtually experiences).
According to \cite{arikan_iti09}, $I(W_N^{(i)})$ (the capacity of bit channel $i$) converges to either 0 or 1 as $N$ tends to infinity, and the fraction of the bit channels with capacity 1 approaches $I(W)$.

With finite block lengths, not all bit channels are fully polarized. The principle of
polar codes is to choose the $K$ most reliable bit
channels among $N$ bit channels to convey information bits. The other bits are called frozen bits which are fixed to
be transmitted on the rest channels. The good information set is denoted as $\mathcal{A}$ and complementary set is
${\mathcal{A}_c}$. Denote $u_{\mathcal{A}}$ as a subvector of the vector $u_1^N$ that takes elements of it from
the set $\mathcal{A}$.

%\subsection{Successive Cancelation Decoding}
The SC decoder is proposed in \cite{arikan_iti09} and it recursively computes the likelihood ratio (LR)
of bit $i$ from
\begin{gather}
L_N^{(i)}(y_1^N,\hat{u}_1^{i-1})\triangleq \frac{W_N^{(i)}(y_1^N,\hat{u}_1^{i-1}|0)}{W_N^{(i)}(y_1^N,\hat{u}_1^{i-1}|1)},
\end{gather}
where $\hat{u}_1^{i-1}$ is the estimation of bits $u_1^{i-1}$. The SC decoder
%observes $(y_1^N,u_{\mathcal{A}^c})$ and
generates the estimate $\hat{u}_i$ of bit $u_i$ ($i \in \mathcal{A}$) from
\begin{gather}
\hat{u}_i = \left\{ \begin{array}{ll}
0, & \textrm{if $L_N^{(i)}(y_1^N,\hat{u}_1^{(i-1)})\geq 1$}\\
1, & \textrm{otherwise}.\\
\end{array} \right.
\end{gather}
The decoding complexity of SC is $\mathcal{O}(N\log N)$ \cite{arikan_iti09}.
%----------------------------------%--------------------------------%

\section{Polar Codes with Memory} \label{sec_polar_memory}
In this section, the encoding of PCM  and the decoding strategies are introduced, which improves the performance of polar codes under the SC, BP or SCL decoding.

\subsection{Encoding with Memory}\label{sec_encoding_memory}
The top-level scheme is shown in Fig.~\ref{fig_system_model}.
Let $K_{\mathrm{crc}}$ denote the number of CRC bits in each block and these CRC bits are part of the $K$ information bits.
Then there are $K_{\mathrm{info}} = K - K_{\mathrm{crc}}$ pure information bits in each block.
Let $K_{\mathrm{p}}$ be the number of the mutual information bits , and the number of the rest information bits is denoted as $K_{\mathrm{i}} = K_{\mathrm{info}} - K_{\mathrm{p}}$.

\begin{figure}
{\par\centering
\resizebox*{3.0in}{!}{\includegraphics{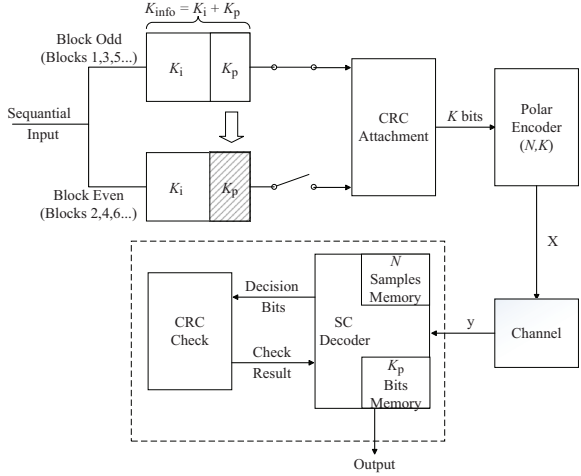}} \par}
\caption{System model of PCM employing the SC decoding. The encoding diagram is shown here
in a sequential manner. A parallel encoding can be done
with two CRC attachment modules and two polar encoders. }
\label{fig_system_model}
\end{figure}

In the encoding process, a frame of sequential input bits is first divided into chunks of the length $2K_{\mathrm{i}}+K_{\mathrm{p}}$.
Then each chunk is divided into two blocks: Block Odd (with $K_{\mathrm{i}} + K_{\mathrm{p}}$ bits) and Block Even (with $K_{\mathrm{i}}$ bits).
Block Even then takes the $K_{\mathrm{p}}$ bits from Block Odd to form an input vector with the length $K_{\mathrm{info}}$ for its CRC generation.
In this way, there are clearly $K_{\mathrm{p}}$ mutual information bits which are both included in Block Odd and Block Even. These mutual information bits are placed at the same indices, and the mutual information set is denoted as $\mathcal{B}$.
The input bit stream arrangement is shown in Fig.~\ref{fig_encoding_detail}.

\begin{figure}
{\par\centering
\resizebox*{3.0in}{!}{\includegraphics{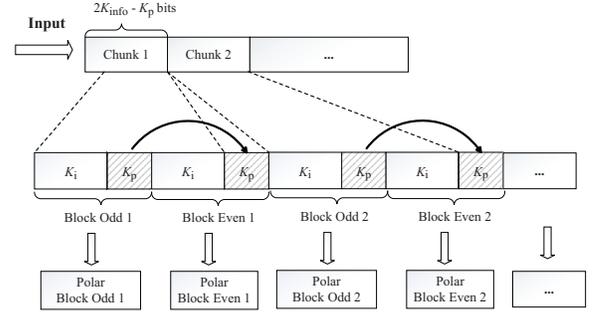}} \par}
\caption{The input bit arrangement of PCM.}
\label{fig_encoding_detail}
\end{figure}

The encoding of the two blocks can be done sequentially or in parallel as seen from Fig.~\ref{fig_system_model}, where
both the CRC attachment and the polar encoding are performed to Block Odd and Block Even alternatively, under the control of a switch.

\subsection{The Decoding Process} \label{sec_decoding_process}
The symbols of encoded code blocks are transmitted over the symmetric B-DMC channel $W$, and the noisy version of them are observed at the receiver side.
The receiver collects chucks of samples with a length of $2N$: the first $N$ samples for Block Odd and the rest for Block Even.
The SC decoder generates an estimate $\hat{u}_1^N$ for each block. The CRC check module returns a check result for each block. The possible check results are:
\begin{itemize}
  \item Case$\,$$1$: Both Block Odd and Block Even are decoded correctly;
  \item Case$\:$$2$: Block Odd is decoded correctly but Block Even is decoded incorrectly;
  \item Case$\:$$3$: Block Odd is decoded incorrectly while Block Even is decoded correctly;
  \item Case$\:$$4$: Both Blocks are decoded incorrectly.
\end{itemize}\par
For Case 1 and Case 2, since Block Odd is decoded correctly, the $K_{\mathrm{p}}$ estimations of the mutual information bits are stored in the memory for possible re-use by the second round of decoding of Block Even. For Case 1, since Block Even is also decoded correctly, there is no need for any more actions. For Case 2, Block Even is decoded incorrectly, a new round of SC decoding for Block Even can be carried out.
%then two types of decoding processing can be carried out for Block Even: a new round of SC decoding for Block Even or a new round of BP decoding for Block Even. With the initial SC decoding, these two processing are called SC-SC decoding and SC-BP decoding, respectively.\par
For Case 3 and Case 4, since Block Odd is decoded incorrectly, the initial $N$ LR values of this block need to be saved for a possible new round of decoding. For Case 3, the correctly decoded Block Even can provide the estimations of the mutual information bits to Block Odd, invoking a new round of SC decoding of Block Odd. For Case 4, since Block Even is also decoded incorrectly, there is nothing the decoder can do for both blocks.
%For Case 3, the correctly decoded Block Even can provide soft or hard decisions for the $K_{\mathrm{p}}$ mutual information bits, depending on the available type processing. For Case 4, since Block Even is also decoded incorrectly, there is no new information fed from it to the previous Block Odd. The decoding process claims fails for both Block Odd and Block Even.
%\subsubsection{SC-SC Decoding}

A more detailed description of the decoding process when a new round of SC decoding occurs is as follows.
%A new round of SC decoding is performed to Block Even (Case 2) or Block Odd (Case 3).
The $K_{\mathrm{p}}$ estimations of the mutual information bits from the correctly decoded block are fed to the incorrectly decoded block. Take Case 2 as an example. Here the decoder of Block Even can repeat the SC decoding up to the first bit in $\mathcal{B}$. When it reaches to the first bit with the index $i \in {\mathcal{B}}$, then the decoder takes this bit as a frozen bit: no matter what the calculated LR value is for $u_i$, it is assigned to the decision taken from Block Odd. The SC decoding process goes on until the end, treating all bits in $\mathcal{B}$ as frozen bits. The re-decoding of Block Odd in Case 3 is the same as that of Case 2.
%\subsubsection{SC-BP Decoding}
%For Case 2 (or Case 3), re-decoding of Block Even  (or Block Odd) can also be carried out using the BP decoding process. This time we take Case 3 as an example, where Block Odd is incorrectly decoded in the SC decoding process. With Block Even correctly decoded, the $K_{\mathrm{p}}$ mutual information bits can be treated as frozen bits with infinite log LR (LLR) values. Then BP decoding can be carried out for Block Odd.

\subsection{Positioning of Mutual Information Bits}\label{sec_mutual_bits}
Every two consecutive transmitting blocks share $K_{\mathrm{p}}$ mutual information bits. The positioning of mutual information bits is to find an optimal way in assigning these mutual information bits to the input of the two blocks (Block Odd and Block Even). Here ``optimal" means the best system error performance. The exact formulation is derived as follows. The size of set $\mathcal{B}$ is $|\mathcal{B}| = K_{\mathrm{p}}$ and the subvector $u_{\mathcal{B}}$ contains the mutual information bits. Theoretically, there are $\binom{K}{K_{\mathrm{p}}}$ ways to choose the set $\mathcal{B}$.
Assume the information set $\mathcal{A} = \{i_1, i_2, ...,i_K\}$ is ordered in the ascending order with respect to the bit channel reliability. In other words, there exists the relationship of $P_e(W_N^{(i_1)}) \ge P_e(W_N^{(i_2)}) \ge ... \ge P_e(W_N^{(i_K)})$, where $P_e(W_{N}^{(i)})$ is the error probability of the $i$-th information bit. The following proposition states an optimal way to achieve the best union bound.
\newtheorem{proposition}{Proposition}
\begin{proposition}\label{pro_position}
Supposing the information set $\mathcal{A} = \{i_1, i_2, ...,i_K\}$ is ordered in the ascending order with respect to the bit channel reliability, then the set $\mathcal{B}$ containing the first $K_{\mathrm{p}}$ elements of the set $\mathcal{A}$ as the mutual information bits indices can produce the minimum union bound.
\end{proposition}
\begin{proof}
%In this work, we propose to place the $K_{\mathrm{p}}$ mutual bits as the most poorly protected information bits.
Define the PER over the information set $\mathcal{A}$ as $P_B(\mathcal{A})$. Then its union bound \cite{vardy_iti21} is
\begin{gather}
P_B(\mathcal{A})\leq\sum_{i\in{\mathcal{A}}}P_e(W_{N}^{(i)}).
\end{gather}
With a pair of consecutive code blocks, when re-decoding is performed for either of them, it is equivalent to the case that the information set
of the other block is $\mathcal{A}'=\mathcal{A} \setminus \mathcal{B}$. This is because one block is decoded correctly and the
mutual information bits are now considered as frozen bits for another block. In such circumstance, the union bound
for the incorrectly decoded block is:
\begin{gather}
P_B(\mathcal{A}') \leq \sum_{i\in{\mathcal{A}'}}P_e(W_{N}^{(i)}).
\end{gather}
Supposing set $\mathcal{B}'$ is any other mutual information set, so the equivalent information set of the incorrect block can be similarly
derived as $\mathcal{A}''=\mathcal{A} \setminus \mathcal{B}'$. So we can get
\begin{gather}
P_B(\mathcal{A}'') \leq \sum_{i\in{\mathcal{A}''}}P_e(W_{N}^{(i)}).
\end{gather}
Because set $\mathcal{B}$ contains the indices corresponding to the $K_{\mathrm{p}}$ largest error probabilities in $\mathcal{A}$, it is obvious that
\begin{gather}
\sum_{i\in\mathcal{B}}P_e(W_N^{(i)}) \ge \sum_{i\in\mathcal{B}'}P_e(W_N^{(i)}).
\end{gather}
Therefore,
\begin{gather}
\sum_{i\in\mathcal{A}'}P_e(W_N^{(i)}) \leq \sum_{i\in\mathcal{A}''}P_e(W_N^{(i)}).
\end{gather}
It means that the union bound of $P_B(\mathcal{A}')$ is smaller than $P_B(\mathcal{A}'')$. Since $\mathcal{B}'$ is arbitrary, we can conclude that $P_B(\mathcal{A}')$ has the smallest union bound.
\end{proof}

\subsection{Error Performance Analysis}\label{sec_error_analysis}
In this section, the error performance of PCM is analyzed. Here we omit the inside argument of $P_B(\mathcal{A})$ for compactness. Instead, the symbol $P_B$ is used to represent the underlying PER of polar codes with the information set $\mathcal{A}$.
The PER of PCM consists of two parts:
\begin{itemize}
\item Part$\,$1: Block Odd  and Block Even are both decoded incorrectly, corresponding to Case 4 in Section \ref{sec_decoding_process}.
\item Part$\;$2: The re-decoding of Block Even (Case 2) or Block Odd (Case 3) fails.
\end{itemize}
For Part 1, the error probability is $P_B^2$. For Part 2, supposing the PER of the re-decoding is $P'_B$, the error probability is therefore $P_B(1-P_B)P'_B$. The PER of PCM is therefore:
\begin{equation}\label{eq_new_per}
P_{new} = P_B^2 + P_B(1-P_B)P'_B.
\end{equation}
With the optimal placement of the mutual information bits in Section \ref{sec_mutual_bits}, there must be some blocks which can be recovered with the help of additional $K_{\mathrm{p}}$ frozen bits. Representing $P'_B$ by $\alpha P_B$, where $\alpha$ can be obtained empirically for now, Eq.~(\ref{eq_new_per}) can be rewritten as:
\begin{equation}\label{eq_new_bound}
P_{new} = P_B^2 + P_B(1-P_B) \alpha P_B = (1+\alpha)P^2_B-\alpha P^3_B.
\end{equation}
By Eq.~(\ref{eq_new_bound}), it is shown that with the same complexity of the SC decoding, PCM can
achieve a PER which is on the order of the underlying PER squared.

%----------------------------%-------------------------%
\subsection{Decoding with a BP or SCL Decoder}\label{sec_bp_decoding}
In the proposed PCM, the SC decoding can be perfectly replaced by the BP or SCL decoding. For Case 2 and Case 3, only one block is decoded correctly.
The correctly decoded block can provide correct decisions of the mutual information bits to be
used by the incorrectly decoded block. Here note that for the BP decoding, the best way to use these correct decisions is to treat the
mutual information bits as frozen bits, instead of using the soft values of them. The reason is simple: by
treating them as frozen bits, the initial LR values of these bits are equivalently set to be infinity, which is
definitely better than using finite soft LR values from the correctly decoded block. As for the SCL decoding, the mutual information bits are treated as frozen bits directly. Therefore, even with the BP or SCL decoding for PCM, the mutual information bits are used in the same way as the PCM employing the SC decoding.

\subsection{Comparison with Turbo Codes}\label{sec_comparision_turbo}
The encoding of PCM shares a certain amount of information bits between a pair of consecutive blocks. This can be compared with Turbo codes with two parallel  identical encoders.
Compared to Turbo codes, there are several differences.

First, all the incoming information bits go through two identical encoders for Turbo codes.
While PCM only shares a fraction of information bits between two encoding blocks, enabling a flexible code rate configuration.
The second difference is that PCM does not constantly
exchange soft information between two blocks in the decoding process.
Instead, only when one block fails and the other succeeds, estimations
of the mutual information bits are fed from the succeeded block to the failed block.
The information pass can be considered as a sporadic procedure: the average percent of
all additional rounds of decoding is only (denoted as $P_a$)
\begin{equation}\label{eq_additional_sc}
P_{a} = P_B(1-P_B) = P_B-P_B^2 < P_B.
\end{equation}
Compared with stand-alone polar codes, the additional decoding can result in a significant reduction in PER as shown in Eq.~(\ref{eq_new_bound}), and it accounts for only a small percentage of the overall decoding operations.

\section{General Polar Codes with Memory}\label{sec_gen_polar}
In Section \ref{sec_polar_memory}, PCM is proposed where two consecutive blocks share a controlled number of information bits. A natural question arises: can we extend this scheme to $m > 2$  polar blocks and possibly achieve a better error performance?
The direct extension of the encoding scheme from two blocks to $m$ blocks ($m>2$) is first analyzed in this section. Then an improved encoding scheme is proposed which achieves the same order of the PER while improves the overall code rate of the direct extension. This improved version is called general PCM in this paper.
\subsection{Direct Extension of Polar Codes with Memory}\label{sec_per_analysis}
The PER of PCM in Section \ref{sec_error_analysis} is $P_B^2 + P_B(1-P_B)P'_B$, where each chunk contains two blocks.
When this scheme is extended to $m$ blocks with each containing $K_{\mathrm{p}}$ mutual information bits, the PER consists of the following parts:
\begin{itemize}
\item Part$\,$1: Only one block is decoded incorrectly, and the new round of the decoding fails again;
\item Part$\;$2: Two blocks are decoded incorrectly, and at least one block in the new round of decoding fails again;
\item ...
\item Part$\;$$m$: All of the $m$ blocks are decoded incorrectly.
\end{itemize}
For Part 1, because the re-decoding of the failed block fails again, there is one block error among $m$ polar blocks. The PER in this case is therefore:
\begin{equation}
P_1 = \frac{1}{m}C_m^1P_B(1-P_B)^{m-1}P'_B.
\end{equation}
For Part 2, with two blocks failed, the final block error among $m$ blocks consists of two case: 1) one of the re-decoded blocks
fails and 2) both of the re-decoded blocks are decoded incorrectly. Therefore, the PER is:
\begin{equation}
P_2 = \frac{2}{m}C_m^2P_B^2(1-P_B)^{m-2}(\frac{1}{2}C_2^1P'_B(1-P'_B)+P_B^{'2}).
\end{equation}
%there is no way the mutual $K_{\mathrm{p}}$ bits of them can be recovered.
%the error probability is
%$\frac{2}{m}C_m^2P_B^2(1-P_B)^{m-2}(P_B^{'2}+P'_B(1-P'_B))$.
Generally, for Part $k$, $(1 \leq k < m)$, the error probability $P_k$ is:
\begin{equation}
\begin{split}
P_k &= \frac{k}{m}C_m^kP_B^k(1-P_B)^{m-k}(\frac{1}{k}C_k^1P'_B(1-P'_B)^{k-1}\\
&+\frac{2}{k}C_k^2P_B^{'2}(1-P'_B)^{k-2}+...+P_B^{'k}).
\end{split}
\end{equation}

For Part $m$, because all of the blocks are decoded incorrectly, the error probability $P_m$ is simply $P_m = P_B^m$.
Accumulating the error probability of each part and simplifying the formula, the PER of the direct extension
of PCM is obtained:

\begin{equation}\label{eq_n_blocks}
\begin{split}
P_{new} &= \sum_{k=1}^mP_k = P_B(1-P_B)^{m-1}P'_B\\
&+\frac{2}{m}C_m^2P_B^2(1-P_B)^{m-2}P'_B+...+ P_B^m.
\end{split}
\end{equation}

Replacing $P'_B$ by $\alpha P_B$ in Eq.~(\ref{eq_n_blocks}), we can obtain a new PER:

\begin{equation}\label{eq_n_blocks_bound}
\begin{split}
P_{new}&= \alpha P_B^2(1-P_B)^{m-1}\\
&+\frac{2}{m}C_m^2\alpha P_B^3(1-P_B)^{m-2}+...+ P_B^m.
\end{split}
\end{equation}

With a relatively small $P_B$, the new PER is dominated by
$\alpha P_B^2(1-P_B)^{m-1}$, which corresponds to the situation when only one block is decoded incorrectly.
All the other parts have terms on the order of at least $P_B^3$.
Based on this fact, a general encoding scheme in next section is proposed to deal with the case where one block is decoded incorrectly among $m$ blocks. For all the other cases, no re-decoding is performed.
This enables the scheme to still maintain the same PER order while improves the overall code rate.

\subsection{The General Polar Codes with Memory}\label{sec_encoding_n_blocks}
In this section, a general encoding scheme of PCM is proposed.
From the discussions in the previous section, it can be seen that the direct extension of the encoding scheme does not increase the minimum order of the PER.
The PER performance is limited by the error event that there is only one failed block among $m$ blocks. All the other error events have lower PER level. If the encoding scheme is designed to only recover the limiting error event while ignoring those error events with lower PER level, then the overall effective code rate can be improved.

For the direct extension of PCM, the effective overall code rate is
\begin{equation}\label{eq_n_mutual}
\begin{split}
R_m &= \frac{m(K-K_{\mathrm{crc}})-(m-1)K_{\mathrm{p}}}{mN} \\
&= R -\frac{K_{\mathrm{crc}}}{N}-\frac{m-1}{m}\frac{K_{\mathrm{p}}}{N},
\end{split}
\end{equation}
with a rate loss of $\frac{K_{\mathrm{crc}}}{N}+\frac{m-1}{m}\frac{K_{\mathrm{p}}}{N}$, where $R$ denotes the code rate of the underlying polar codes.
To reduce the rate loss of the direct extension of PCM, a general encoding scheme is proposed.
Fig.~\ref{fig_n_block_encoding_detail} shows such an input bit arrangement of the general PCM, where each chunk contains $m$ blocks.
In Fig.~\ref{fig_n_block_encoding_detail}, the first $m-1$ blocks have their own information bits, no mutual information bits are shared among them. However, for each of these $m-1$ blocks, $K_{\mathrm{p}}$ information bits are taken out and added together (modulo two addition).
The resultant $K_{\mathrm{p}}$ bits are put as the mutual information bits for the last block. So the input bit arrangement of the general PCM can be shown as follows:
\begin{equation}\label{eq_n_mutual}
u_{\mathcal{B}}^m = u_{\mathcal{B}}^1 \oplus u_{\mathcal{B}}^2 \oplus ... \oplus u_{\mathcal{B}}^{m-1},
\end{equation}
where $u_{\mathcal{B}}^k$, $k\in(1,2,...,m)$ denotes the $K_{\mathrm{p}}$ mutual information bits of block $k$.
The positioning of the mutual information bits for all $m$ blocks follows Proposition~\ref{pro_position}: they are put as those most poorly protected information bits in each block.

In this way, the effective code rate of the general PCM is:
\begin{equation}\label{eq_n_coderate}
R_m = \frac{m(K-K_{\mathrm{crc}})-K_{\mathrm{p}}}{mN} = R-\frac{K_{\mathrm{crc}}}{N}-\frac{K_{\mathrm{p}}}{mN}.
\end{equation}
With a large $m$, the fractional rate loss $K_{\mathrm{p}}/mN$ is negligible with a constant $K_{\mathrm{p}}$.
However, a large $m$ comes with a higher decoding complexity. Trade-off can always be made between a small rate loss and a lower decoding latency.

The design of the general encoding scheme can recover the $K_{\mathrm{p}}$ mutual information bits of the failed block if all other $m-1$ blocks in the chunk are decoded successfully: $u_{\mathcal{B}}^k = \sum_{i=1, i\neq k}^{m}u_{\mathcal{B}}^{i}$. This scheme can not correct more than one block error among $m$ blocks. When there is only one incorrectly decoded block, the correct $K_{\mathrm{p}}$ mutual information bits can be recovered and a new round of decoding can be performed.

For the general PCM, the new round of decoding occurs only when one block is decoded incorrectly, the PER of our proposed scheme is:
\begin{equation}\label{eq_n_blocks_proposed}
\begin{split}
P_{new} &= P_B(1-P_B)^{m-1}P'_B +\frac{2}{m}C_m^2P_B^2(1-P_B)^{m-2}\\
&+...+\frac{k}{m}C_m^kP_B^k(1-P_B)^{m-k}+... + P_B^m,
\end{split}
\end{equation}
which can be rewritten by replacing $P'_B$ by $\alpha P_B$:
\begin{equation}\label{eq_n_blocks_bound_proposed}
\begin{split}
P_{new}&= \alpha P_B^2(1-P_B)^{m-1}+\frac{2}{m}C_m^2P_B^2(1-P_B)^{m-2}\\
&+...+\frac{k}{m}C_m^kP_B^k(1-P_B)^{m-k}+... + P_B^m.
\end{split}
\end{equation}
Comparing Eq.~(\ref{eq_n_blocks_bound}) and Eq.~(\ref{eq_n_blocks_bound_proposed}), the general PCM scheme has negligible performance loss compared with the direct extension scheme.

\begin{figure}
{\par\centering
\resizebox*{3.0in}{!}{\includegraphics{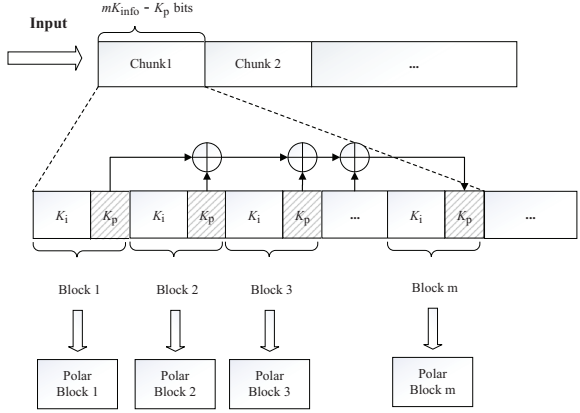}} \par}
\caption{The input bit arrangement of the general PCM.}
\label{fig_n_block_encoding_detail}
\end{figure}

\section{Simulation Results} \label{sec_numerical}
In this section, we provide simulation results to show the performance of PCM. The channel is the
additive white Gaussian noise (AWGN) channel. The block length of the polar codes is $N=256$, and the number of underlying information bits is $K=140$, including a 12-bit CRC with a generator polynomial $g(x)=x^{12}+x^{11}+x^{10}+x^9+x^8+x^4+x+1$.
The code rate of the underlying polar codes is therefore $R =\frac{K}{N} = 0.5469$. The number of mutual information bits shared between two consecutive blocks is set as $K_{\mathrm{p}}=24$.

Fig.~\ref{fig_ber_r_0p5} reports the BER performance of the PCM with two consecutive blocks sharing mutual information bits.
The effective code rate of the PCM-SC-2 is $R_2 = R-\frac{2(K-K_{\mathrm{crc}})-K_{\mathrm{p}}}{2N}=R-\frac{K_{\mathrm{crc}}}{N}-\frac{K_{\mathrm{p}}}{2N}=0.4531$.
For a fair comparison, the code rate of the stand-alone polar codes with SC, BP, and SCL decoding is adjusted as $R_2$, and the stand-alone polar codes with SC and BP decoding also contain a 12-bit CRC.
For the stand-alone SCL decoding, the list size $L$ is simulated for both $L=2$ and $L=4$.
It is observed that the PCM-SC-$2$ outperforms the traditional SC and BP decoding by about 0.41 dB and 0.22 dB at BER=$10^{-4}$, respectively.
In addition, PCM-SC-2 achieves a comparable performance (less than 0.3 dB) as the SCL decoding with $L=2$ at the same BER level.
On the other hand, the PCM-BP-$2$ achieves the same performance as the SCL decoding with $L=2$ when $E_b/N_0 \ge 4$  dB.
Fig.~\ref{fig_per_r_0p5} shows the corresponding PER performance, and the trend is consistent with that shown in Fig.~\ref{fig_ber_r_0p5}.
%The curve labeled as $P_{new}$ shown in this figure is the new PER derived from Eq.~(\ref{eq_new_bound}), where $\alpha$ is set as 1.7 at $E_b/N_0<3.5$ dB and 6.9 at $E_b/N_0 \geq 3.5$ dB. It is observed that $P_{new}$ curve is on the top of the PCM-SC-2, which indicates that the PER performance of the PCM-SC-2 is on the level of PER square of the underlying polar codes.

Fig.~\ref{fig_bound} shows the simulated PER of the PCM-SC-2 and the PER analyzed in Eq.~(\ref{eq_new_bound}). Here the maximum (6.9) and the minimum (0.38) values of $\alpha$ are found from the simulations, producing the $P_{new}^{upper}$ and $P_{new}^{lower}$ in Fig.~\ref{fig_bound}. It is observed that the PER performance of the PCM-SC-2 follows the lower bound for small $E_b/N_0$ values (less than 3 dB), and it follows the upper bound for large $E_b/N_0$ values (larger than 3 dB), which indicates that the PER performance of the PCM-SC-2 is on the level of PER squared of the underlying polar codes.

Fig.~\ref{fig_per_PCM_SCL} reports the PER performance of the PCM employing the SCL decoding. It is shown that the PCM-SCL-2 with $L=2$ achieves the same performance as the stand-alone SCL decoding with $L=4$ when $E_b/N_0 > 3.5$ dB. And the PCM-SCL-2 with $L=4$ and $L=8$ outperform the stand-alone SCL decoding with $L=8$ and $L=16$ by about 0.1 dB and 0.15 dB at PER=$10^{-4}$, respectively.

\begin{figure}
{\par\centering
\resizebox*{3.0in}{!}{\includegraphics{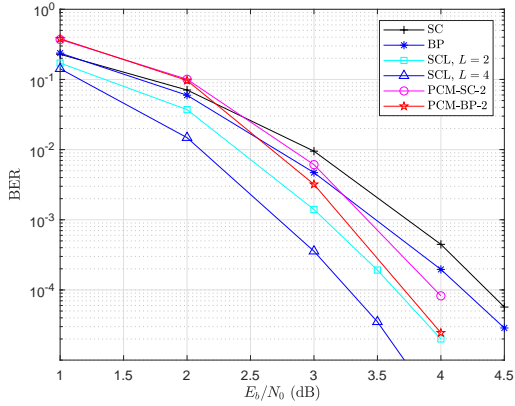}} \par}
\caption{BER of the PCM-SC-2 and PCM-BP-2 over AWGN channels: the block length is $N=256$ and each block contains $K=140$ information bits, including 12-bit CRC and $K_{\mathrm{p}}=24$ mutual information bits. The code rate of the underlying polar codes is $R=0.5469$.
The stand-alone polar codes with SC, BP, and SCL decoding all have adjusted code rate of $R_2=0.4531$.}
\label{fig_ber_r_0p5}
\end{figure}

\begin{figure}
{\par\centering
\resizebox*{3.0in}{!}{\includegraphics{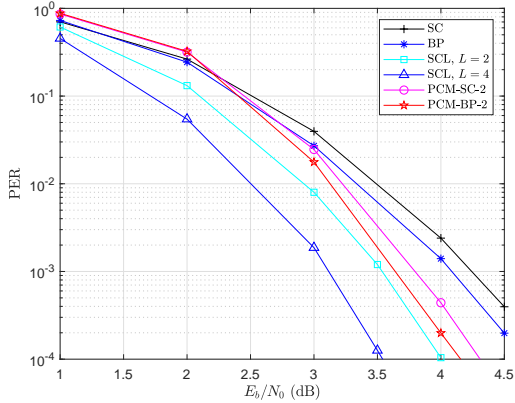}} \par}
\caption{ PER of the PCM-SC-2 and PCM-BP-2 over AWGN channels. The parameters are the same as those in Fig.~\ref{fig_ber_r_0p5}.}
\label{fig_per_r_0p5}
\end{figure}

\begin{figure}
{\par\centering
\resizebox*{3.0in}{!}{\includegraphics{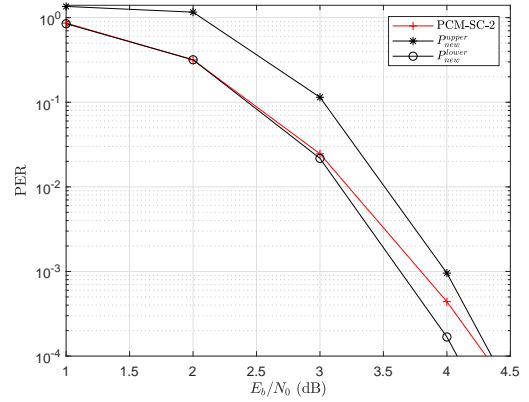}} \par}
\caption{The simulated PER of the PCM-SC-2 and the analyzed PER in Eq.~(\ref{eq_new_bound}). The parameters are the same as those in Fig.~\ref{fig_ber_r_0p5}. The $\alpha$ is set as $\alpha=6.9$ for the $P_{new}^{upper}$ and $\alpha=0.38$ for the $P_{new}^{lower}$.}
\label{fig_bound}
\end{figure}

\begin{figure}
{\par\centering
\resizebox*{3.0in}{!}{\includegraphics{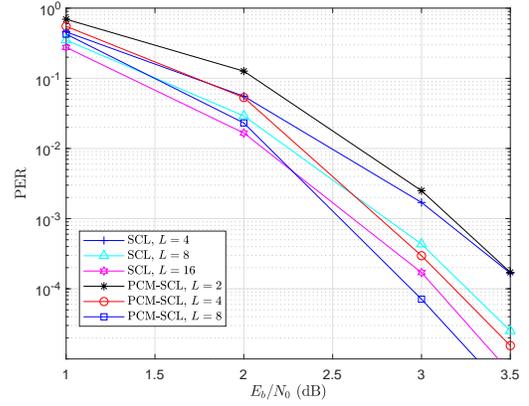}} \par}
\caption{ PER of the PCM-SCL-2 over AWGN channels. The parameters are the same as those in Fig.~\ref{fig_ber_r_0p5}.}
\label{fig_per_PCM_SCL}
\end{figure}

The good performance of PCM comes at an additional round of decoding, as
shown by Eq.~(\ref{eq_additional_sc}). Fig.~\ref{fig_additional_decoding_rate} shows the ratio of the additional
decoding to the overall decoding for the same system in Figs.~\ref{fig_ber_r_0p5} and \ref{fig_per_r_0p5}.
The curve labeled as $P_B$ shown in this figure is the PER of the underlying polar codes.
The success rate of the additional decoding is also provided in this figure, shown by the line with asterisks.
%Also shown in this
%figure is the success rate from the additional SC decoding efforts (shown by the line with asterisks).
It can be seen that for $E_b/N_0 \ge 3$ dB, the additional decoding rate and the additional success rate are matched.
What is important is that the additional decoding efforts are
controlled by the PER of the underlying polar codes: the $P_B$ curve also matches closely with the other two curves for large $E_b/N_0$.
This can be seen from Eq.~(\ref{eq_additional_sc}): when $P_B$ is small ($1-P_B$ approaches 1), the additional decoding rate
$P_a$ is determined by $P_B$. The decoding failure rate of PCM is therefore left with an order of $P_B^2$.

Fig.~\ref{fig_3_blocks_results} presents the PER curves of the general PCM with $m=3$, and the parameters are the same as those in Fig.~\ref{fig_ber_r_0p5}.
By applying the proposed general encoding scheme, the PER of PCM-SC-3 can be represented as follows:
\begin{equation}\label{eq_3_blocks_bound}
P_{new} = (2+\alpha)P_B^2-(1+2\alpha)P_B^3+\alpha P_B^4.
\end{equation}
According to Eq.~(\ref{eq_n_coderate}), the code rate of PCM-SC-$3$ is $R_3=0.4688$, so the stand-alone polar codes with SC, BP, and SCL decoding all have the same adjusted code rate as $R_3$. It can be seen that PCM-SC-3 is about 0.18 dB worse compared with the stand-alone SCL with $L=2$ at PER=$10^{-3}$ level.

\begin{figure}
{\par\centering
\resizebox*{3.0in}{!}{\includegraphics{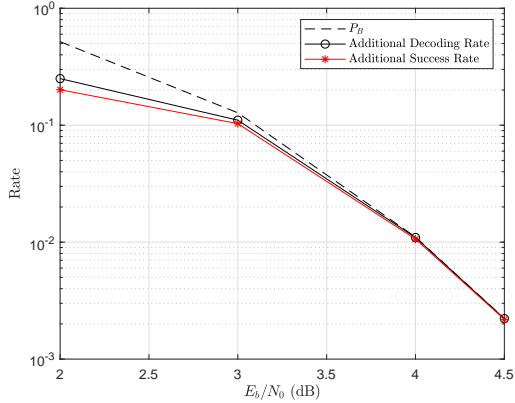}} \par}
\caption{The additional decoding rate of the PCM-SC-2 over AWGN channels.
The parameters are the same as those in Fig.~\ref{fig_ber_r_0p5}. The curve with legend $P_B$ is the PER
of the underlying polar codes.}
\label{fig_additional_decoding_rate}
\end{figure}
\begin{figure}
{\par\centering
\resizebox*{3.0in}{!}{\includegraphics{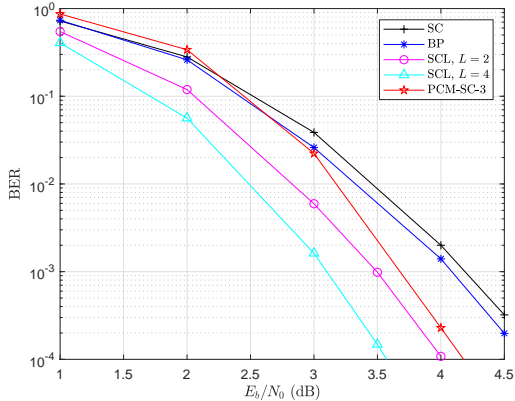}} \par}
\caption{PER of the PCM-SC-3 over AWGN channels. The parameters are the same as those in Fig.~\ref{fig_ber_r_0p5}. The effective code rate of the PCM-SC-3 is $R_3=0.4688$, and the stand-alone polar codes with SC, BP, and SCL decoding all have the same code rate $R_3$.}
\label{fig_3_blocks_results}
\end{figure}

\section{Hardware Architecture}\label{sec_hardware}
In this section, two hardware architectures for PCM-SC-2 are proposed---the IS architecture and the LLI architecture. The IS architecture is based on the SC decoder proposed in \cite{zhangchuan01}, where the processing elements (PEs) are designed with pre-computation. The proposed architecture is capable of performing both SC decoding and PCM-SC-2 decoding. The LLI architecture is inspired by the 2-interleaved SC polar decoder \cite{zhangchuan02}, and it can reduce decoding latency remarkably with only a small increase in hardware consumption compared with the IS architecture.

\subsection{In-serial PCM-SC-2 Decoder}
\begin{figure}
{\par\centering
\resizebox*{3.0in}{!}{\includegraphics{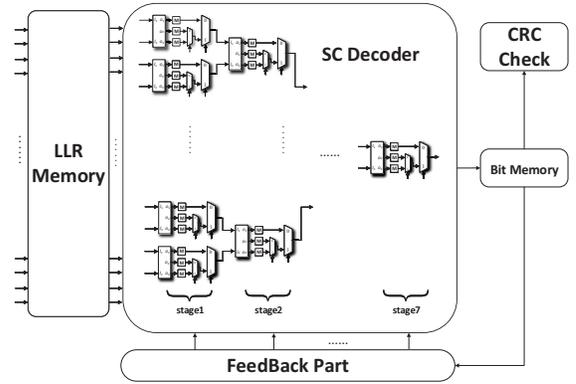}} \par}
\caption{IS PCM-SC-2 decoder.}
\label{simple_SCM}
\end{figure}

In order to increase hardware utilization and reduce computational complexity, the decoder processes the data in the form of log-likelihood ratio (LLR) instead of LR. The top-level architecture of the proposed IS PCM-SC-2 decoder is shown in Fig.~\ref{simple_SCM}. It mainly consists of five modules: the LLR memory module, the SC decoder module, the CRC check module, the feedback module, and the bit memory module. Compared with the conventional SC decoder \cite{zhangchuan01}, the LLR memory module and the bit memory module are additional.

The LLR memory is used to store LLRs which are needed for Case 2 and Case 3. The bit memory module is an important module in the architecture. In the conventional SC decoder, the location and the content of frozen bits are set in advance. When the bit memory receives a frozen bit, it neglects it, and this frozen bit is sent to the feedback module directly. In the PCM-SC-2 decoder, when Block Odd and Block Even are decoded in the first round, they are decoded in the same way as in a conventional SC decoder. When Block Odd or Block Even passes the CRC check, the bit memory immediately stores the mutual information bits of this block. When it comes to Case 2 and Case 3, the bit memory will read mutual information bits and treat them as frozen bits in the second round decoding of the failed block. In this way, the $K_{\mathrm{p}}$ mutual information bits estimates from the correctly decoded block are effectively fed to the bit memory of incorrectly decoded block.

\subsection{Low-latency Interleaved PCM-SC-2 Decoder}
It should be noticed that when it comes to Case 2 and Case 3, the decoding latency of the IS PCM-SC-2 decoder is $1.5$ times of the conventional SC decoder. When Block Odd or Block Even performs a new round of decoding, the computation of LLRs is redundant before the first erroneous mutual information bit. Based on this, an LLI architecture employing interleaved decoding is proposed and shown in Fig.~\ref{highSpeed_SCM}, which is introduced as follows.

\begin{figure}
{\par\centering
\resizebox*{3.0in}{!}{\includegraphics{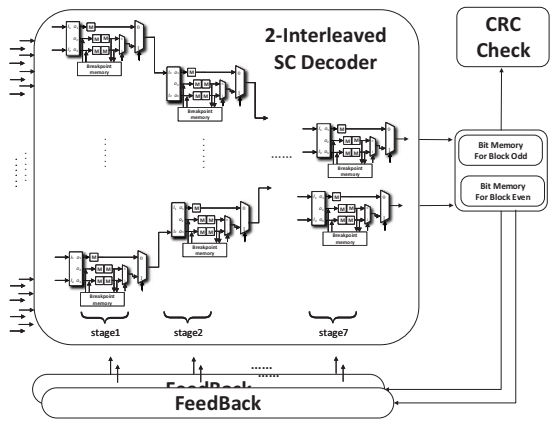}} \par}
\caption{LLI PCM-SC-2 decoder.}
\label{highSpeed_SCM}
\end{figure}

\begin{figure}
{\par\centering
\resizebox*{3.0in}{!}{\includegraphics{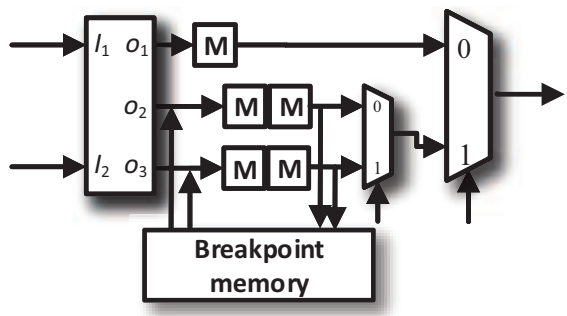}} \par}
\caption{Processing element of LLI architecture.}
\label{process_element}
\end{figure}

\begin{table*}[htbp]
  \centering
  \footnotesize
  \tabcolsep 4mm
  \renewcommand{\arraystretch}{1.5}
  \caption{SYNTHESIS RESULTS COMPARISON OF DIFFERENT POLAR DECODERS FOR $N=256$.}
\begin{tabular}{l|cc|c|cc}
\Xhline{1.0pt}
\multirow{2}{2.2cm}{\textbf{Decoder}}
&\multicolumn{2}{c|}{\textbf{PCM-SC-2 Decoder}}
&\multirow{2}{2.2cm}{\textbf{Combinational\\SC Decoder} \cite{Dizdar01}}
&\multicolumn{2}{c}{\textbf{Adaptive SCL Decoder} \cite{Sural01}}\\
\cline{2-3} \cline{5-6}
~ &\textbf{IS} &\textbf{LLI} &~ &\textbf{$L = 2$} &\textbf{$L = 4$} \\ \hline
\rowcolor{tabgray}  \textbf{LUTs}    & $4302$    & $4781$    & $35152$  &$12589$    & $16565$   \\ \hline
                     \textbf{FFs}     & $4629$    & $15381$    & $1561$  &$7809$    & $10217$   \\ \hline
\rowcolor{tabgray}  \textbf{Total}   & $8931$    & $20126$    & $36944$  &$20398$   &$26782$   \\ \hline
                     \textbf{RAM [bit]} &0           &0           &$1792$    &0     &0      \\ \hline
\rowcolor{tabgray}  \textbf{Block RAMs} & 0         &0          &0           &14   &28     \\ \hline
\textbf{Max. Freq. [MHz]}           &105.3           &104.7               &---        &135.5  &121.5 \\ \hline
\rowcolor{tabgray} \textbf{Min.-Max.Latency [$\mu$s]}    &$2.44-2.95$      &$1.23-1.82$       &---      &$1.12-11.32$     &$1.25-12.67$   \\ \hline
\textbf{Min.-Max.T/P [Mbps]}  & $83-100$      & $166-201$    &---  &$12-197$  & $11-177$  \\
\Xhline{1.0pt}
\end{tabular}%
\label{table:polar_decoders_comparison}%
\end{table*}%
For a conventional $N$-bit SC decoder, there are $n$ stages---Stage 1 to Stage $n$, with only one stage being active in a clock cycle. As described in \cite{zhangchuan02}, a 2-interleaved SC decoder can decode two polar blocks simultaneously. Inspired by this, the LLI architecture is proposed. The main idea is that when Block Odd is being decoded in Stage $i$ $(i\neq n)$, Block Even can be decoded in Stage $i-1$ since this stage is idle for Block Odd. The decoding process of the two blocks will conflict in the last stage---Stage $n$, because every block needs to stay in Stage $n$ for two clock cycles. Therefore, an additional PE is needed in this stage.
As shown in Fig.~\ref{highSpeed_SCM}, LLI PCM-SC-2 decoder has an extra PE in the Stage $n$. In addition, two independent bit memories and feedback modules are designed for Block Odd and Block Even, in order to decode them simultaneously.
Fig.~\ref{process_element} shows the PE of the LLI PCM-SC-2 decoder. It has two additional registers which are used to store LLRs of the two blocks, compared with that of the IS PCM-SC-2 decoder.
%As shown in Fig.~\ref{process_element}, the 2-interleaved PCM-SC-2 decoder adds two registers in every PE which are used to store LLRs of the two blocks and inserts an extra PE in Stage $n$.

With the proposed design, whenever a mutual information bit in Block Even is decoded, it can be compared with the mutual information bit of the same location in Block Odd (which was decoded one clock cycle before). If the two bits are different, all intermediate LLR values of the two blocks, which are stored in the registers, will be immediately sent to the breakpoint memory, and then the decoding process continues. When it comes to Case 2 or Case 3, the incorrectly decoded block starts the second round of decoding from the position of the first different mutual information bit, and the intermediate LLRs are directly read from the breakpoint memory instead of being calculated. In the studied case of PCM with the same parameters as those of Fig.~\ref{fig_ber_r_0p5}, the indices of the first mutual bit and the last mutual bit are 32 and 209, respectively. It means that if the second round of decoding is required, the computation of the LLRs before the $32$-th bit is avoided in the worst case, and the computation of the LLRs before the $209$-th bit is avoided in the best case.

\subsection{Implementation Results}
\begin{figure}
{\par\centering
\resizebox*{3.0in}{!}{\includegraphics{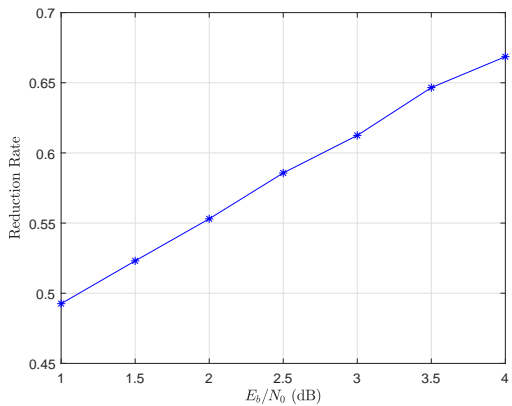}} \par}
\caption{Reduction rate of the average latency for the LLI PCM-SC-2 decoder in the second round of decoding.}
\label{reduction_latency}
\end{figure}

\begin{figure}

{\par\centering
\resizebox*{3.0in}{!}{\includegraphics{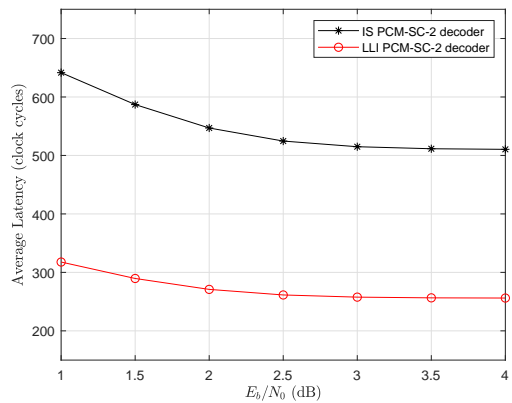}} \par}
\caption{Average latency of the IS PCM-SC-2 decoder and the LLI PCM-SC-2 decoder (in clock cycles).}
\label{average_latency}
\end{figure}
The two decoders are implemented on the Xilinx ZNYQ-7000 field-programmable gate array (FPGA) platform.
The latency of the LLI decoder is lower than that of the IS decoder, and it is reduced nearly by half in the first round of decoding, due to the interleaved decoding for Block Odd and Block Even. It is also reduced in the second round of decoding, because the LLI decoder can begin with the first erroneous mutual information bit.
Fig.~\ref{reduction_latency} shows the reduction rate of the average latency for the LLI PCM-SC-2 decoder in the second round of decoding, and the number of the samples is 100000 at each $E_b/N_0$. It is remarkable that the LLI PCM-SC-2 decoder can reduce the latency of the second round decoding by $49.3\%$ and $66.9\%$ at $E_b/N_0=1$ dB and $E_b/N_0=4$ dB, respectively. Fig.~\ref{average_latency} shows the average latency of the two decoders.
It can be seen that the average latency of the LLI PCM-SC-2 decoder is approximately half of that of the IS PCM-SC-2 decoder because the latency reduction rates are both around $50\%$ in the first round and the second round.

Table~\ref{table:polar_decoders_comparison} shows the synthesis results comparison of different polar decoders for $N=256$, including the IS and LLI PCM-SC-2 decoders, the combinational SC decoder in \cite{Dizdar01}, and the adaptive SCL decoder \cite{Sural01} with $L=2$ and $L=4$.
As shown in Table~\ref{table:polar_decoders_comparison}, the IS PCM-SC-2 decoder consumes the least hardware resources, although the maximum throughput is inferior to the others. The total consumptions (FF and LUT) of the IS PCM-SC-2 decoder and the LLI PCM-SC-2 decoder are only $24.2\%$ and $54.5\%$ of the consumption of SC decoder in \cite{Dizdar01}, respectively. The hardware consumption of the LLI PCM-SC-2 decoder is $98.7\%$ and $75.2\%$ of the consumption of adaptive SCL decoder with $L=2$ and $L=4$, respectively. The reason is that the SCL decoder needs $L$ SC decoder modules while the proposed architecture only needs one. Moreover, the decoders in \cite{Sural01} and \cite{Dizdar01} use additional RAM and Block RAMs, which increases the consumption of hardware resources, while the proposed PCM-SC-2 decoders do not.

Table~\ref{table:polar_decoders_comparison} also shows the range of the latency and the throughput of the PCM-SC-2 decoders and the adaptive SCL decoder. It is observed that the minimum latency and the maximum throughput of the LLI PCM-SC-2 decoder are comparable to those of the adaptive SCL decoder with $L=2$, and are slightly superior to those of the adaptive SCL decoder with $L=4$.
As for the worst situation, the maximum latency of the LLI PCM-SC-2 decoder is only $16.1\%$ and $14.4\%$ of the adaptive SCL decoder with $L=2$ and $L=4$, and the minimum throughput is improved by more than 13 and 15 times compared to them, respectively.

\section{Conclusion}\label{sec_con}
In this paper, PCM employing the SC, BP, or SCL decoding is proposed. By sharing a certain amount of mutual
information bits between a pair of blocks,  this scheme can bring down the PER to the square of the
underlying polar codes. Results show that for the block length 256, the proposed PCM-SC-2 and PCM-BP-2 decoders can match the PER of the stand-alone SCL decoder with two lists. The PER performance of PCM-SCL-2 decoder with $L$ lists can match the PER of the stand-alone SCL decoder with $2L$ lists. In the meantime, the proposed LLI hardware architecture for PCM can achieve 13 times more throughput compared to the adaptive SCL decoder with two lists when the block length $N=256$ in the worst case.
%\IEEEtriggeratref{7}
\bibliography{ref_polar}
\bibliographystyle{IEEEtran}
\end{document}